\documentclass{article}
\RequirePackage{fix-cm}
\usepackage{graphics,graphicx}
\usepackage{xcolor}
\usepackage{hyperref}
\usepackage{listings}
\usepackage{algorithm}
\usepackage[noend]{algpseudocode}
\usepackage{xspace}
\usepackage{amsmath,amssymb}
\usepackage{multirow}
\usepackage[round]{natbib}

\usepackage{authblk}

\newtheorem{definition}{Definition}
\newtheorem{theorem}{Theorem}
\newtheorem{proof}{Proof}

\newcommand{\review}[1]{{{#1}}}

\newcommand{\ER}{Erd\H{o}s-Rényi\xspace}

\def \calS{{\cal S}}

\def \nat{\mathbb{N}}

\def \GT{\mathcal{G}}
\def \GS{G}
\def \Ne{\mathcal{N}}
\def \dt{\Delta t}
\def \snaps{G_{t_1,..,t_m}}
\def \id{\textproc{id}}

\newcommand{\en}[2]{\ensuremath{G_{#2}(#1)}\xspace}
\newcommand{\etn}[2]{\ensuremath{G_{#2}^{k}(#1)}\xspace}
\newcommand{\etns}[2]{\ensuremath{s_{#2}^{k}(#1)}\xspace}

\newcommand{\etmm}{\textproc{ETMM}\xspace}

\newcommand{\tmm}{\textproc{TMM}\xspace}

\newcommand{\etnemb}[2]{\ensuremath{EMB_{#2}(#1)}\xspace}
\newcommand{\etndist}[3]{\ensuremath{dist_{#3}(#1,#2)}\xspace}
\newcommand{\etmdist}[2]{\ensuremath{dist(#1,#2)}\xspace}

\begin{document}

\title{An Efficient Procedure for Mining Egocentric Temporal Motifs}

\author[1,2]{Antonio Longa}
\author[1]{Giulia Cencetti}
\author[1]{Bruno Lepri}
\author[2]{Andrea Passerini}
\affil[1]{Fondazione Bruno Kessler (FBK), Trento - Italy}
\affil[2]{University of Trento, Trento - Italy}


\maketitle

\begin{abstract}
Temporal graphs are structures which model relational data between entities that change over time. Due to the complex structure of data, mining statistically significant temporal subgraphs, also known as temporal motifs, is a challenging task. In this work, we present an efficient technique for  extracting temporal motifs in temporal networks. Our method is based on the novel notion of egocentric temporal neighborhoods, namely multi-layer structures centered on an ego node. Each temporal layer of the structure consists of the first-order neighborhood of the ego node, and corresponding nodes in sequential layers are connected by an edge. 
The strength of this approach lies in the possibility of encoding these structures into a unique bit vector, thus bypassing the problem of graph isomorphism in searching for temporal motifs. This allows  our algorithm to mine substantially larger motifs with respect to alternative approaches.
Furthermore, by bringing the focus on the temporal dynamics of the interactions of a specific node, our model allows to mine temporal motifs which are visibly interpretable. Experiments on a number of complex networks of social interactions confirm the advantage of the proposed approach over alternative non-egocentric solutions. The egocentric procedure  is indeed more efficient in revealing similarities and discrepancies among different social environments, independently of the different technologies used to collect data, which instead affect standard non-egocentric measures.\\
\textbf{keywords:}Network motifs, Temporal networks,  Graph mining, Social interaction networks, Sociopatterns
\end{abstract}

\section{Introduction}
\label{intro}

\begin{figure}[h!]
\centering
    \includegraphics[width=0.75\linewidth]{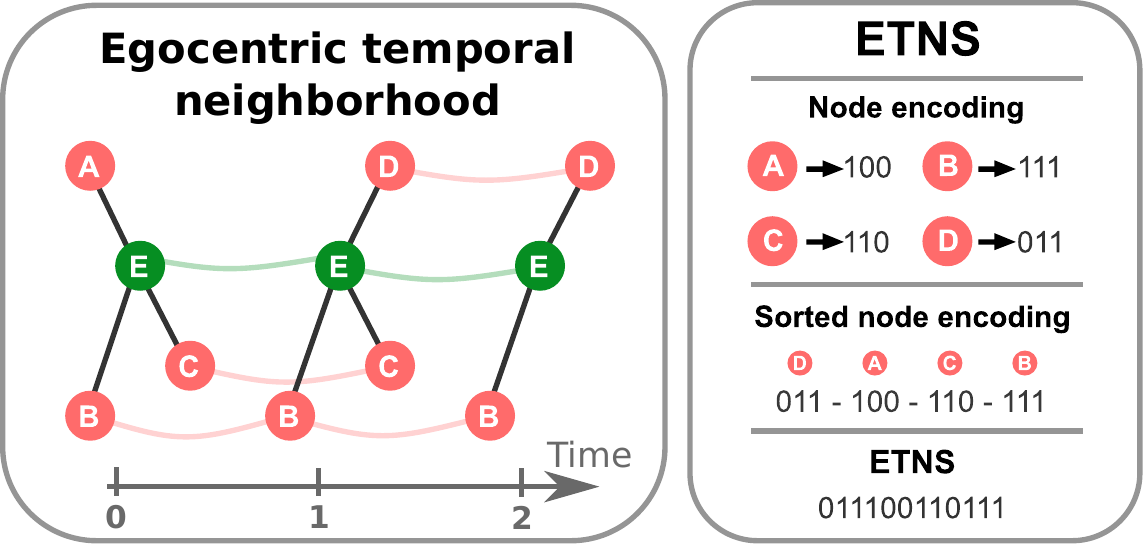}
    \caption{\label{fig:overview} Graphical summary of the procedure for
      extracting egocentric temporal motifs.
      The left panel shows the {\it egocentric temporal neighborhood} of the ego node $E$ (in green), with temporal order two and initial time instant zero. Black edges connect the central node with its neighbors (in red) at each time step, while green (resp. red) edges connect consecutive occurrences of the central
      (resp. a neighboring) node along the time sequence. The right panel shows how the corresponding {\it egocentric temporal neighborhood signature} (ETNS) is computed. Each neighboring node is encoded into a bit vector indicating the time slots when it is present. 
      The node encodings are lexicographically sorted first and then concatenated to generate the signature.
      }
\end{figure}

Complex networks play a pivotal role in describing and analyzing complex systems in multiple natural and artificial scenarios, representing a fundamental tool for modeling biological, cognitive and social systems \cite{newman2010}.
Interestingly, the small substructures that compose the complex topology of a network are sometimes recurrently emerging as essential constituents for the specific network at hand. They  consist in sub-networks composed by a small number of nodes with a specific structure of connections. 
The substructures which are identified as the most significant take the name of \textit{motifs}~\cite{milo2002network,alon2007network}. The significance of each specific substructure within the overall network architecture is assessed in relation to its frequency  and  usually  referring to a null model: a structure is considered a motif if the number of its occurrences in the network is substantially higher than the number of occurrences in the null model.\\
The identification of specific repeated motifs offers a unique
opportunity to investigate the complex and intricate dynamics of human behavior and interactions~\cite{wasserman1994social,milgram1967small}. As a matter of fact, when analyzing social dynamics we usually need to deal with time-dependent structures~\cite{kossinets2006,kossinets2008}. 
Social interactions are indeed characterized by links which appear and disappear in time and are associated with variable duration. The appropriate topological tool to describe systems of dynamical interactions is represented by temporal networks with a fixed set of nodes connected by edges that vary over time \cite{holme2012}.  In such framework the identification of motifs becomes more challenging, since a substructure can be repeated both in time and in space.  A vast literature addresses the definition of \textit{temporal motifs} and the ways to identify them~\cite{jazayeri2020motif,kovanen2011temporal,ray2014frequent,zhao2010communication,gurukar2015commit,nicosia2013graph,kosyfaki2018flow,jin2007trend}.
Inspired by the work of Paranjape \textit{et al.}~\cite{paranjape2017motifs}, the aim of this paper is to further extend the concept of temporal motifs going beyond the traditional point of view.  The standard approach is indeed based on observing temporal networks from the outside and decomposing them in their small components. The idea of our approach is instead to jump inside the network and follow the path of a specific node, finding node-dependent spatio-temporal patterns. 
In particular, for each node we observe its neighbors and how its connections to them change in a given period of time. We neglect the connections among neighbors of the chosen ``ego" node, and we only focus on studying how the set of neighbors evolves in time, following an \textit{ego perspective}.
In social settings this allows to identify the patterns of interactions of individuals, selecting the most relevant behaviors as those which are most repeated in time by the same or different persons. We give to these ones the name of \textit{egocentric temporal motifs (ETM)}.
\review{The ego perspective allows to address the motif identification procedure very efficiently by
comparing egocentric temporal sub-networks in terms of their {\em signature}, simply consisting of a bit vector. This represents a huge simplification with respect to mining standard motifs, which necessarily requires to  address the graph isomorphism problem, which slows down the procedure and makes it hard to identify graph motifs with more than a handful of nodes.}
A graphical summary of our approach is shown in Figure~\ref{fig:overview}.

We conducted an extensive experimental evaluation applying our mining algorithm to a number of diverse interaction datasets. First, we analyzed a set of close proximity interaction networks, including three high schools, a hospital, a research institute, a primary school and a university campus. Qualitative results indicate that, as compared to non-egocentric alternatives, egocentric temporal motifs are more intuitive and representative of the differences between these environments and the categories of the underlying egos. Quantitative results show that a metric based on egocentric temporal motifs is more effective than existing micro-scale, meso-scale and global-scale alternatives in discriminating between different types of graphs. Second, we studied the ability of egocentric temporal motifs to discriminate distance communication networks based on the technology employed (phone calls, sms or emails) and to distinguish different types of synthetic networks (i.e., temporal variants of \ER, scale-free and small-world networks). Results confirm the effectiveness and generality of the egocentric perspective in characterizing a wide range of interactions and highlight the conditions under which this perspective can be limiting.

\section{Related work}\label{related_work}

In the last years, a number of solutions has been developed
to mine motifs in temporal networks (see~\cite{jazayeri2020motif} for a survey). 
In this manuscript we focus on temporal networks where nodes are fixed and edges can change over time. 
Currently, two popular strategies have been followed to adapt graph mining approaches to deal with a changing network topology. The first strategy~\cite{araujo2014com2,dunlavy2011temporal,tantipathananandh2007framework} consists of aggregating temporal information, i.e., building a static network
containing all connections in the temporal graph regardless of the
time associated to them. While this simple strategy allows to use
standard techniques for motifs discovery, it loses the ability to
capture the temporal dynamics of the interactions between nodes. The
second one consists in building a growing network, where nodes
and edges can be added but never deleted~\cite{ray2014frequent,leskovec2007graph}.
However these approaches are not appropriate to deal with data containing
social interactions which are necessarily transient.

Most methods for mining motifs of
transient interactions have been developed in the field of
communication networks. Kovanen \textit{et al.}~\cite{kovanen2011temporal} define
the concept of $\dt-connected$ graph as the connected temporal graph containing edges 
within a temporal gap $\dt$, and search for temporal motifs inside them. Zhao \textit{et
al.} \cite{zhao2010communication} extended this concept to
\textit{communication motifs}, basically requiring a number of occurrences greater than a given threshold.

Later, Gurukar \textit{et al.} \cite{gurukar2015commit} proposed COMMIT, an algorithm that converts connected temporal subgraphs in sequences using graph invariants and then mines frequent
sub-sequences as communication motifs. 
More recently, Kosyfaki \textit{et al.} \cite{kosyfaki2018flow} proposed a new definition of max-flow communication motifs, in which flow refers to data (e.g., money, messages, etc.). 
Hulovatyy \textit{et al.} \cite{hulovatyy2015exploring} introduced \textit{dynamic graphlets}, which extend the concept of graphlets  from static networks to temporal graphs. However they do not search for temporal motifs, but rather use all dynamic graphlets (up to a given complexity) to generate vectorial representations of the network and its nodes. \review{A related line of research aims at characterizing temporal networks in terms of dense subgraphs~\cite{kostakis2017discovering,rozenshtein2017finding,rozenshtein2020mining}.} Finally, Paranjape \textit{et
al.} \cite{paranjape2017motifs} propose a mining strategy that extracts
static motifs from the aggregate network (obtained collapsing all the temporal layers together and thus dropping the temporal information) and expands them into temporal motifs
by considering the order of appearance of edges within a given temporal
gap. Other studies investigated  approximate methods for counting temporal motifs \cite{liu2019sampling,wang2020efficient}.\\
None of these approaches tries to capture the temporal evolution of the interactions of a single node, which is the focus of our work. The egocentric perspective allows to extract meaningful patterns of interaction that are hard to find with non-egocentric solutions. Additionally, it allows to devise an efficient procedure to compare these types of patterns that can substantially speed up the mining process.

\section{Mathematical Background}\label{background}

\begin{definition}[\textit{Graph}]
  A graph $G$ can be defined as a pair $(V,E)$, where $V$ is a set of
  vertices or nodes, and $E$ is a set of edges between the nodes,
  i.e., $E \subseteq \{(u,v) | u,v \in V\}$. 
\end{definition}

\begin{definition}[\textit{Graph isomorphism}]\label{def:iso}
  Two graphs $G = (V,E)$ and $G' = (V',E')$ are said to be isomorphic
  if and only if there exists a bijection $\pi$ between their vertex
  sets such that for all $(u,v) \in E$ it holds that
  $(\pi(u),\pi(v)) \in E'$ (edge-preservation). Graph isomorphism is denoted as $G \simeq G'$.
\end{definition}

\begin{definition}[\textit{Node neighborhood}]
  Given a  graph $G = (V,E)$, the {\em neighbors} of a node
  $v \in V$ are the set of nodes adjacent to $u$, i.e.,
  $\Ne(v) = \{u \in V | (u,v) \in E\}$. The node {\em neighborhood} is
  the subgraph of $G$ containing $v$ and its neighbors as nodes and
  all edges connecting them as edges.
\end{definition}

As previously stated, network motifs are patterns of connections
occurring on a given network significantly more often than in random
networks~\cite{milo2002network}. The next definition formalizes the
concept.

\begin{definition}[\textit{Network motif}]\label{def:milo}
Given a graph $G$ and a set of $n$ random graphs $G_{0}$, a sub-graph $M$ of $G$ is a network motif if and only if: $(i)$ $Pr(\Bar{N}_{G_0} > N_{G}) < \alpha$  (\textit{over-representation}); $(ii)$ $N_{G} - \Bar{N}_{G_0} \ge \beta \Bar{N}_{G_0}$  (\textit{minimum deviation}); $(iii)$ $N_{G} \ge \gamma$ (\textit{minimum frequency}).
Here $N_{G}$ is the number of occurrences of sub-graph $M$ in $G$,
$\Bar{N}_{G_0}$ is the average number of occurrences of sub-graph $M$ in the random graphs ($G_0$) and $\alpha \in [0,1]$, $\beta \in [0,1]$ and
$\gamma \in \nat$ are parameters.

\end{definition}

The \textit{over-representation} condition requires that the probability of observing a motif in the random graphs more than in the original one is lower than a certain threshold $\alpha$. \textit{Minimum deviation}  instead prevents the detection as motifs of subgraphs  with a slight difference in occurrences between the graph under investigation and the random graphs. Finally,
\textit{minimum frequency} avoids detecting statistically significant but
infrequent motifs.

\begin{figure*}[h!]
\centering
    \includegraphics[width=\linewidth]{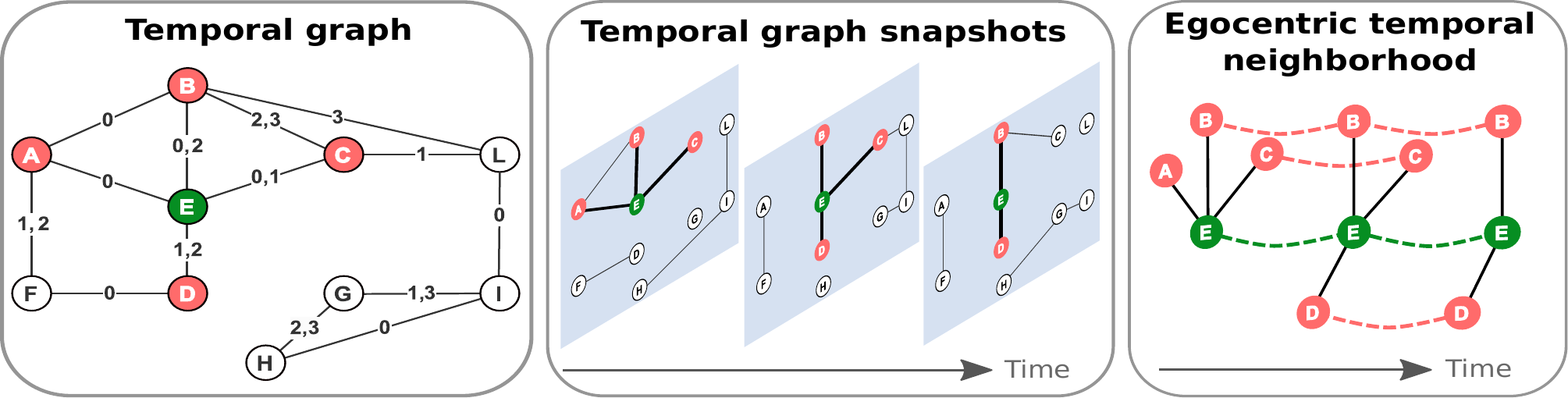}
    \caption{The left panel shows a temporal graph $\GT$ focused on a
      node $E$. The middle panel shows three graph snapshots, and the right panel shows a $k=2$ order ETN for $E$ built
      from the sequence of its egocentric neighborhoods.}
    \label{fig:etn}
\end{figure*}

\begin{definition}[\textit{Temporal graph}]\label{def:tg}
  A temporal graph $\GT=(V,E)$ is a pair of sets where $V$ is a set of
  vertices or nodes and $E$ is a set of {\em temporal} edges, i.e.,
  edges enriched with temporal information. Each temporal edge $e \in E$ is a
  quadruple $(u,v,t_{start},t_{end})$, where $u$ and $v$ are nodes
  ($u,v \in V$) and $t_{start}$ and $t_{end}$ are time instants representing, respectively, the beginning and the end of the
  interaction between node $u$ and node $v$. Given a temporal graph
  $\GT$, its corresponding (static) {\em aggregate} graph $\GS$ is obtained
  removing temporal information from the edges of $\GT$.
\end{definition}

\begin{definition}[\textit{Temporal graph snapshot}]\label{def:tg_snap}
  Given a temporal graph $\GT=(V,E)$ and a temporal gap $\dt$, a
  temporal graph snapshot at time $t$ is a static graph
  $G_t = (V_t,E_t)$ such that $V_t = V$ and 
  there is a static edge $(u,v) \in E_t$ if and only if the corresponding temporal interaction $(u,v,t_{start},t_{end}) \in E$ exists within $\dt$, i.e. $ t_{start} \in [t,t + \dt) \lor t_{end} \in [t,t + \dt) $.

  A temporal graph $\GT=(V,E)$ can be represented as a sequence of
  temporal graph snapshots $G_{t_1},G_{t_2},...,G_{t_m}$ where $t_1$
  is the smallest $t_{start}$ in $E$, $t_{i+1} = t_i + \dt$ and $t_m$
  is smaller than the largest $t_{end}$ in $E$.
\end{definition}

\section{Mining egocentric temporal motifs}\label{ETM}

\begin{figure}[h!]
\centering
  \includegraphics[width=0.7\linewidth]{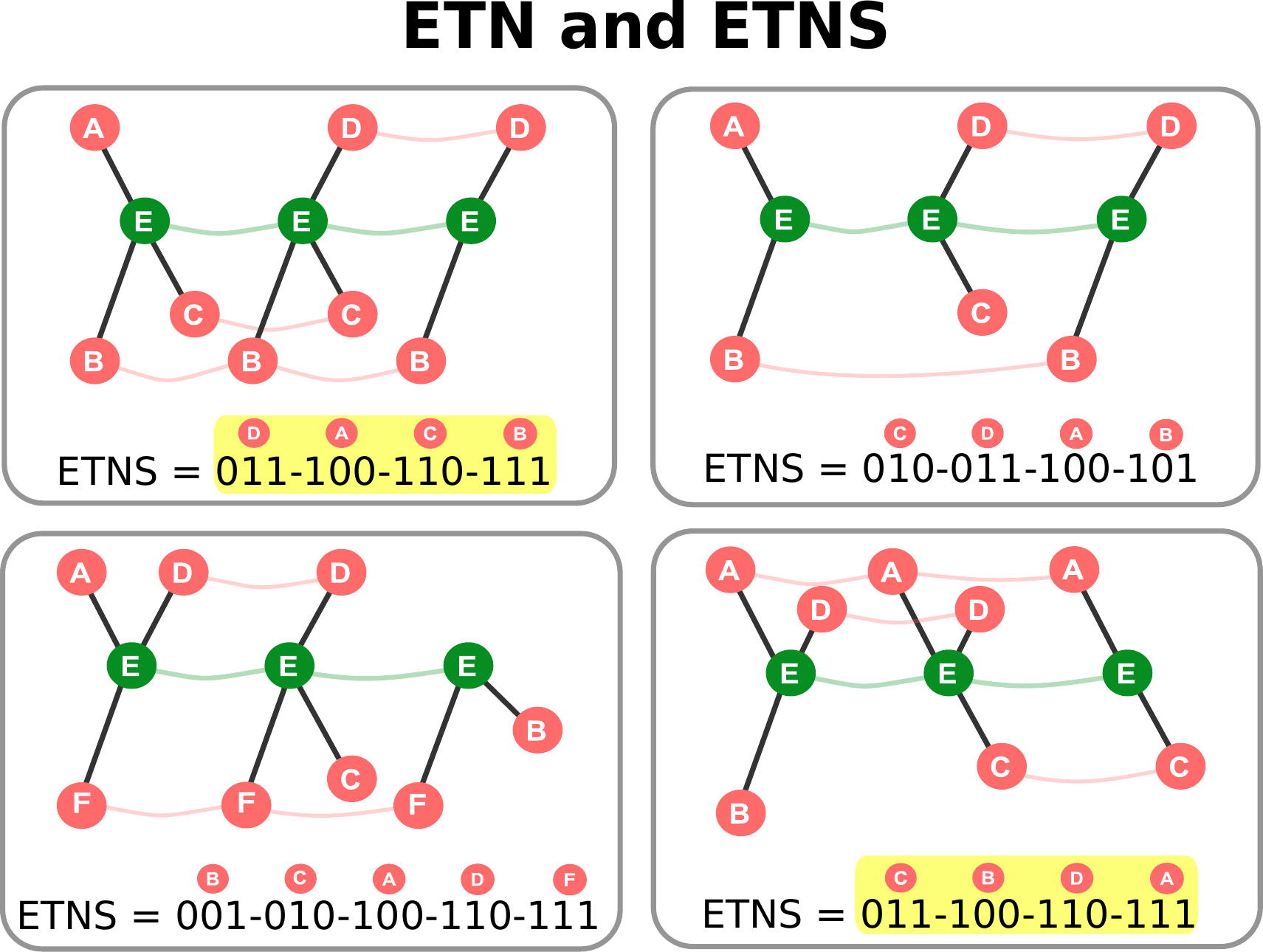}
  \caption{\label{fig:etns_examples} Examples of ETN and ETNS for
    different temporal graphs with $k=2$. The two highlighted ETNS are
    identical and correspond to isomorphic ETN.}
\end{figure}

Let us start by introducing the notions of {\em egocentric neighborhood} and {\em egocentric temporal neighborhood}.

\begin{definition}[\textit{Egocentric neighborhood}]
  Given a (static) graph
  $G = (V,E)$ and a node $v \in V$, the {\em egocentric neighborhood}
  of $v$ is the subgraph $G(v)$ obtained by taking the
  neighborhood of $v$ and removing all edges not including $v$ as one   of the nodes.
\end{definition} 

Note that this simple variant of the node neighborhood focuses the attention on the central node, dropping all information not related to it. We next show how to extend this egocentric focus to the temporal aspect,
by following the temporal evolution of the node neighborhood.

\begin{definition}[\textit{Egocentric temporal neighborhood -- ETN}]\label{def:etn}
  Given a temporal graph $\GT = (V,E)$, a temporal gap $\dt$, a temporal neighborhood order $k$ and a node $v \in V$, the {\em egocentric temporal neighborhood} of $v$ is defined as follows. Let $G_{t_1},G_{t_2},...,G_{t_m}$ be the sequence of
  temporal graphs' snapshots for $\GT$ with gap $\dt$. Let $\en{v}{t_1},\en{v}{t_2},\dots,\en{v}{t_m}$ be the sequence of egocentric
  neighborhoods of $v$ for such temporal graph snapshots. The $k$-th
  order egocentric temporal neighborhood of $v$ at time $t_i$ is a
  graph obtained taking $\en{v}{t_i},\dots,\en{v}{t_{i+k}}$ and
  connecting each node to the next occurrence of the same node (if
  any) along the sequence. In addition, each node is labelled with its
  position in the sequence. We refer to this graph as $\etn{v}{t_i}$.
\end{definition} 


Figure~\ref{fig:etn} shows the extraction of an ETN from a temporal
graph.  The structure of ETN graphs allows to efficiently compute
graph isomorphism via a graph signature. To simplify the presentation
of the signature generation algorithm, we assume a function $\id$ that
applied to a node in an ETN returns its identifier in the original
temporal graph $\GT$ (the letters in Figure~\ref{fig:etn}).

\begin{definition}[\textit{Egocentric temporal neighborhood signature (ETNS)}]
  Given a temporal graph $\GT = (V,E)$ and an egocentric temporal
  neighborhood graph $\etn{v}{t}$ for node $v$, time $t$ and order
  $k$, an egocentric temporal neighborhood signature $\etns{v}{t}$ is
  a bit vector encoding $\etn{v}{t}$. Two egocentric temporal
  neighborhoods $\etn{v}{t}$ and $\etn{v'}{t'}$ have the same
  signature if and only if they are isomorphic.
\end{definition}

The procedure computing the ETNS for a given ETN graph is shown in
Algorithm~\ref{alg:etns}. The algorithm starts by initializing the
signature $s$ to an empty vector and collecting all nodes of the ETN
graph with distinct identifiers into a set $V$. Here $V_{t+i}(v)$
indicates the set of nodes in the $t+i$ temporal slice of
$\etn{v}{t}$, and the union discards duplicates according to $\id$. For
each node $u$, with the exception of the central node $v$, the algorithm then computes a
bit vector encoding $s_u$. The encoding has length $k$ and contains at
each position $i$ a Boolean flag stating whether the node (represented
by its identifier $\id$) is present in the corresponding temporal
slice, i.e., $u \in V_{t+i}(v)$. After computing this bit vector, the
algorithm appends it to $s$.  
\review{Finally, the list of neighborhood node signatures is sorted in lexicographic order and concatenated into the final signature.}
Figure~\ref{fig:etns_examples} shows some examples of ETN and
corresponding ETNS for $k=2$.

\begin{algorithm}
\begin{algorithmic}[0]
\Procedure{computeETNS}{$\etn{v}{t}$}
\State $s \gets [\,]$
\State $V \gets \bigcup\limits_{i=0}^{k} V_{t+i}(v)$
\For{$u \in V$}
\If {$u \ne v$}
\State $s_u \gets [\,]$
\For{$i = 0,..,k$}
\If {$u \in V_{t+i}(v)$}
\State \textproc{append}($s_u,1$)
\Else
\State \textproc{append}($s_u,0$)
\EndIf
\EndFor
\EndIf
\textproc{append}($s,s_u$)
\EndFor
\State $s \gets$ \textproc{sort}($s$)
\State {\bf return} \textproc{flatten}($s$)
\EndProcedure
\end{algorithmic}
\caption{\label{alg:etns}Procedure for computing the signature of an ETN graph.}
\end{algorithm}

\begin{theorem}[\textit{Isomorphic ETN}]
  Given two egocentric temporal neighborhoods $\etn{v}{t}$ and
  $\etn{v'}{t'}$, Algorithm~\ref{alg:etns} returns the same signature
  if and only if they are isomorphic.
\end{theorem}

\begin{proof}
  We first show that if two ETNs are isomorphic they have the same
  signature. Let $\pi$ be a bijection for the two ETNs as from
  Definition~\ref{def:iso}. Note that this bijection will map central
  nodes to central nodes\footnote{Apart for the degenerate case
    consisting of a single neighbor running all along the sequence,
    where there is no distinction between central node and neighbor
    and the proof is trivial.} (they are the only ones that can have a
  degree larger than one on a given temporal slice). By specifying a
  mapping between nodes, $\pi$ also implicitly defines a mapping
  between node identifiers. The edge-preserving property of $\pi$
  implies that the mapping of identifiers is consistent (if two
  non-central nodes share an edge they have the same identifier). It
  also implies that the two paired node identifiers share the same set
  of edges, and thus have the same encoding. Having the same encodings
  for each pair of node identifiers, the resulting signatures are also
  the same. This concludes the first part of the proof.

  We next show that if the signatures are the same the ETNs are
  isomorphic. We prove this by showing how to create the bijection
  function $\pi$. Recall that a signature is a flattened sorted list
  of encodings of node identifiers, and that all encodings have the
  same length $k+1$. We start by pairing node identifiers in the two
  graphs by their positions in the respective signatures. We then map
  nodes with paired identifiers by matching their node labels (i.e.,
  positions in the underlying graph sequence). Given that the node
  encodings of the paired identifiers are the same, the corresponding
  nodes appear in the same positions in the underlying sequence (thus
  matching by node labels produces a perfect match). We repeat the
  same matching for the only unpaired node identifiers, which
  correspond to the central node. Note that by definition of ETN the
  central node appears with all labels from 1 to $k+1$. Being
  redundant we omit its encoding from the signature. By construction,
  mapped nodes share the same label, i.e., $\ell(u) = \ell(\pi(u))$
  for all $u$. Concerning edges, by definition of ETN edges are only
  between the central node and the neighbors, and between consecutive
  instances of the same node along the sequence. The former
  requirement is easily satisfied as each (non-central) node is always
  connected to the central node having the same label. The latter is
  satisfied because by construction if two node identifiers have the
  same encoding their corresponding nodes have the same edges (recall
  that central nodes have the same encoding even if it is not part of
  the signature). This concludes the proof.  
\end{proof}

We are now ready to introduce the algorithm for extracting statistics
on ETNs from a temporal graph. The pseudocode of the algorithm is
shown in Algorithm~\ref{alg:motifs}.

\begin{algorithm}
\begin{algorithmic}[0]
\Procedure{countETN}{$\GT$,$\dt$,$k$}
\State $\calS \gets \emptyset$
\State $\snaps \gets$ \textproc{ExtractSnapshots}($\GT$,$\dt$)
 \For{$i = 1,..,m-k$}
 \For{$v \in V_{t_i}$}
 \State $\etn{v}{t_i} \gets$ \textproc{buildETN}($\en{v}{t_i},\dots,\en{v}{t_{i+k}}$)
 \State $\etns{v}{t_i} \gets$ \textproc{computeETNS}($\etn{v}{t_i}$)
 \If{$\etns{v}{t_i} \in \calS$}
 \State $\calS[\etns{v}{t_i}] \gets \calS[\etns{v}{t_i}] + 1$
 \Else 
 \State $\calS[\etns{v}{t_i}] \gets 1$
 \EndIf
 \EndFor
 \EndFor 
 \State {\bf return} $\calS$
 \EndProcedure
\end{algorithmic}
\caption{\label{alg:motifs}Procedure for extracting counts of ETN graphs from a temporal graph.}
\end{algorithm}

The algorithm takes as input a temporal graph $\GT$, a temporal gap
$\dt$ and a temporal neighborhood order $k$ and returns a dictionary
of counts $\calS$ mapping ETNs to the number of occurrences of the
corresponding ETN in $\GT$. It starts by initializing $\calS$ to the
empty set and extracting the sequence of temporal graph snapshots of
$\GT$ for gap $\dt$. For each time $t_i$ and node $v$ ($V_{t_i}$ is
the set of nodes of graph $G_{t_i}$) it builds the corresponding ETN
and computes its associated signature using
Algorithm~\ref{alg:etns}. The signature is finally used to update the
ETN counts in $\calS$. Note that this update step is extremely
efficient thanks to the fact that ETNs are bit vectors.

\begin{theorem}[\textit{Complexity of \textproc{computeETNS}}]
  The worst-case complexity of \textproc{computeETNS} is
  $\mathcal{O}(d^{(k)}\log{d^{(k)}})$, where $d^{(k)}$ is the
  maximal degree of the network when considering edges within a
  $k \cdot \dt$ temporal range.\label{the:comp1}
\end{theorem}

  \begin{proof}
    Building the signature requires to create an encoding of length
    $k+1$ for each of the nodes in $\etn{v}{t_i}$ with distinct
    identifier, which are $|V|$. The complexity is thus
    $\mathcal{O}(|V|)$. Sorting the signature requires sorting
    each of these encodings, which costs
    $\mathcal{O}(|V|\cdot\log{|V|})$. The worst case complexity can
    be obtained setting $|V|=d^{(k)}$, giving
    $\mathcal{O}(d^{(k)}\log{d^{(k)}})$.
  \end{proof}

  \begin{theorem}[\textit{Complexity of \textproc{countETN}}] The
    worst-case complexity of \textproc{countETN} is
    $\mathcal{O}(n \cdot m \cdot d^{(k)}\log{d^{(k)}})$,
    where $n$ is the number of nodes in the network, $m$ is the
    overall number of temporal snapshots, and $k$ and $d^{(k)}$ are
    as in Theorem~\ref{the:comp1}. The number of temporal snapshots is
    computed as $m = (T_{end} - T_{start})/\dt$, where $T_{start}$ and
    $T_{end}$ are the smallest $t_{start}$ and the largest $t_{end}$
    in the network respectively and $\dt$ is the temporal gap.
\end{theorem}

\begin{proof}
  Note first that the procedure \textproc{ExtractSnapshots} is
  introduced to simplify the explanation, but the underlying algorithm
  never explicitly materializes the sequence of temporal graph
  snapshots for the whole network but directly extracts the ETN using
  \textproc{buildETN}. This latter procedure costs $|\etn{v}{t}|$,
  i.e., the number of nodes in the resulting ETN, which is upper
  bounded by $d^{(k)} \cdot k$. The procedure is repeated
  $n \cdot (m-k)$ times. Computing all ETNs thus costs
  $\mathcal{O}(n \cdot m \cdot d^{(k)} \cdot k)$, and converting them
  to ETNs costs
  $\mathcal{O}(n \cdot m \cdot d^{(k)}\log{d^{(k)}})$.  The count
  update can be done in constant time thanks to the fact that ETNs are
  bit vectors, so the overall worst-case complexity is
  $\mathcal{O}(n \cdot m \cdot d^{(k)}\log{d^{(k)}})$.
\end{proof}

Note that for reasonable values of $k$ and $\dt$, $d^{(k)}$ is
independent of the size of the network, so that the overall
complexity is $\mathcal{O}(n \cdot m)$.

To extract statistically significant ETN from a temporal graph $\GT$, we rely on the support of a null model $\Bar{\GT}$, defined as follows~\cite{holme2012,holme2015modern,jazayeri2020motif}:
\review{
\begin{definition}[\textit{Temporal Graph Null Model}]\label{def:null_model}
Given a temporal graph $\GT$, consider the temporal graph snapshot $G_{t_1},G_{t_2},\dots,G_{t_m}$ (Definition \ref{def:tg_snap}) representation of $\GT$. The null model $\Bar{\GT}$ of $\GT$ is obtained by randomly shuffling the snapshots $G_{t_1},G_{t_2},\dots,G_{t_m}$. 
\end{definition}
}

\review{
Hence a null model $\Bar{\GT}$ is a temporal graph with the same number of nodes, the same number of snapshots and the same number of connections between each couple of nodes but without any temporal correlation. The procedure can be repeated an arbitrary number of times to produce a set of null models that the original temporal graph can be compared with.
}\\
As will be shown in the experimental evaluation, this allows to identify non-trivial temporal structures in a much more selective
way with respect to alternative non-egocentric mining approaches. 

Finally, we define the \textit{Egocentric Temporal Motifs (ETM)} as follows:

\review{
\begin{definition}[\textit{Egocentric Temporal Motifs (ETM)}]\label{etm}
Given a temporal graph $\GT$, $n$ null models $\Bar{\GT}$, and the parameters $\alpha$ (\textit{over-representation}), $\beta$ (\textit{minimum deviation}) and $\gamma$ (\textit{minimum frequency}) appearing in  Definition \ref{def:milo}, the set of ETMs for $\GT$ is obtained applying Definition \ref{def:milo} to $\GT$ where sub-graphs are represented by the set of its ETNs found according to Definition~\ref{def:etn} for each of its nodes.
\end{definition}
}

We name the algorithm extracting ETM from a temporal graph ETMM, standing for Egocentric Temporal Motif Miner.

\section{ETM-based graph distance}

To show the importance of the egocentric perspective in networks of social interactions, we introduce a simple metric that measures the distance between graphs in terms of their respective ETM. To do this, we first define the ETN-based embedding of a temporal graph.

\begin{definition}[\textit{ETN-based embedding}]\label{def:etmm_emb}
Given a temporal graph $\GT$ and a list $M$ of ETNs, we define $\etnemb{\GT}{M}$ as the embedding of $\GT$ in a vector of cardinality $|M|$, in which the $i^{th}$ element of $\etnemb{\GT}{M}$ represents the number of occurrences of $M[i]$ in $\GT$.
\end{definition}

Given a list of ETN, the distance between two temporal graphs is then defined as the distance between their respective ETN-based embeddings.

\begin{definition}[\textit{ETN-based distance}]\label{def:etn_dist}
Given two temporal graphs $\GT_1$, $\GT_2$ and a list $M$ of ETNs, we define $\etndist{\GT_1}{\GT_2}{M}$ as the cosine distance between the ETN-based embeddings of $\GT_1$ and $\GT_2$:
\begin{equation}
    \label{eq:etndist}
    \etndist{\GT_1}{\GT_2}{M} = 1 - \dfrac{\etnemb{\GT_1}{M} \cdot \etnemb{\GT_2}{M}}{||\etnemb{\GT_1}{M}|| \; ||\etnemb{\GT_2}{M}||}
\end{equation}
where $\cdot$ is the dot product and $||\cdot||$ is the Euclidean norm.
\end{definition}

The distance between two temporal graphs can now be computed by first extracting their respective lists of ETM, finding the set of ETM shared by the two graphs and computing their ETN-based distance using this set.  

\begin{definition}[\textit{ETM-based distance}]\label{def:etm_dist}
Given two temporal graphs $\GT_1$, $\GT_2$, two corresponding sets of $n$ null models $\Bar{\GT_1}$ and $\Bar{\GT_2}$ and three parameters $\alpha$, $\beta$ and $\gamma$, we define $\etmdist{\GT_1}{\GT_2}$ as:
\begin{equation}
    \label{eq:etmdist}
    \etmdist{\GT_1}{\GT_2} = \etndist{\GT_1}{\GT_2}{M_{1,2}}
\end{equation}
where $M_{1,2} = M_1 \cap M_2$ and $M_1$ (resp. $M_2$) is the list of ETM obtained applying Definition~\ref{etm}  to $\GT_1$ (resp. $\GT_2$).
\end{definition}

\section{Experimental setup}\label{Experimental_evaluation}

\review{In the following we describe the different groups of network datasets we employed in our experiments and the non-egocentric miners and graph distances that we used as competitors.}

\review{\subsection{Close proximity interaction datasets}
The first group of datasets focuses on close proximity interactions, and contains} three \textit{high school} datasets, a \textit{work place}, a \textit{hospital}, a \textit{primary school} and a university campus (\textit{DTU}).
All datasets except the university campus have been collected using the wearable sensors developed by the SocioPatterns\footnote{http://www.sociopatterns.org/} collaboration, equipped with radio-frequency identification devices (RFIDs) capturing face-to-face interactions. The devices record an interaction if and only if there is at least one exchanged signal within 20 seconds, so 20 seconds the smallest time resolution. The DTU dataset~\cite{sapiezynski2019interaction} instead  represents proximity interactions among university students, collected over a  month using Bluetooth technology to infer physical co-location. The datasets are briefly described in the following.\\
\textbf{HighSchool11~\cite{10.1371/journal.pone.0107878}.} The dataset has been
collected in 2011 in Lycée Thiers, Marseilles, France, over four days
(Tuesday to Friday). It reports the interactions
among 118 students and 8 teachers in three different high school classes.  Number of edges: $1709$, number of nodes: $126$.\\
\textbf{HighSchool12~\cite{10.1371/journal.pone.0107878}.} The dataset has been
collected in 2012 in Lycée Thiers, Marseilles, France, over seven days
(Monday to Tuesday of the following week). It reports the interactions
among 180 students in five different high school classes. 
Number of edges: $2220$, number of nodes: $180$.\\
\textbf{HighSchool13~\cite{mastrandrea2015contact}.} The dataset has been
collected in 2013 in Lycée Thiers, Marseilles, France, over five days
in December. It reports the interactions
among 327 students in nine different high school classes.
Number of edges: $5818$, number of nodes: $327$.\\
\textbf{InVS13~\cite{genois2015data}.} The dataset has been collected in 2013 at the \textit{Institut National de Veille Sanitaire}, a health research institute near Paris, over two weeks. The dataset contains 92 individuals divided in five departments: DISQ, DMCT, SFLE, DSE and SRH. Number of edges: $755$, number of nodes: $92$.\\
\textbf{LH10~\cite{vanhems2013estimating}.} The dataset has been
collected in the geriatric ward of a university
hospital~\cite{vanhems2011risk} in Lyon, France, over four days in December 2010.
The individuals belong to  four classes: medical doctors (MED), paramedical staff (NUR), administrative staff (ADM) and patients (PAT). Number of edges: $1139$, number of nodes: $75$.\\
\textbf{Primary school~\cite{stehle2011high}.} The dataset has been
collected in a primary school in France, over two days in October 2009.
The individuals belong to two classes: teachers (10 individuals) and children (232 individuals). 
Number of edges: $8317$, number of nodes: $242$.\\
\textbf{DTU~\cite{sapiezynski2019interaction}.} The dataset represents the interactions among  university freshmen students in the Copenhagen University. In particular, DTU represents the network of interactions among students collected over a month using Bluetooth technology to infer physical proximity.
Number of edges: $79530$, number of nodes: $692$.

\review{
\subsection{Distance communication datasets}
\label{sec:nonphysical_data}
The second group of datasets contains distance interactions with different communication technologies, namely phone calls, SMSs and emails. The idea is to check whether ETMs are capable of distinguishing graphs according to the underlying communication technology. The datasets are briefly described in the following. 
\\
\textbf{DTU calls~\cite{sapiezynski2019interaction}.} The dataset represents phone calls among university freshmen students in the Copenhagen University. Number of edges: $605$, number of nodes: $525$.\\
\textbf{Friends and Family calls~\cite{aharony2011social}.} The dataset represents phone calls among members of a young-family residential living community adjacent to a major research university in North America.
Number of edges: $432$, number of nodes: $129$.\\
\textbf{DTU SMS~\cite{sapiezynski2019interaction}.} The dataset represents SMSs among  university freshmen students in the Copenhagen University. 
Number of edges: $697$, number of nodes: $568$.\\
\textbf{Friends and Family SMS~\cite{aharony2011social}.} The dataset represents SMSs among members of a young-family residential living community adjacent to a major research university in North America.
Number of edges: $153$, number of nodes: $85$.\\
\textbf{Email EU~\cite{paranjape2017motifs}} The dataset is a collection of emails between members of a European research institution.
Number of edges: $16064$, number of nodes: $986$.\\
\textbf{Email DNC~\cite{nr}} The dataset is a collection of leaked emails between members of the 2016 Democratic National Committee.
Number of edges: $16064$, number of nodes: $986$.\\
}

\review{
\subsection{Synthetic datasets}
\label{sec:synth_data}
The last group of datasets consists in synthetic temporal networks, and aims at checking whether ETMs retain information concerning (temporal variants of) popular network topologies. Each network is built as a temporal graph where the first timestamp is a static synthetic network suitably generated, while  the following temporal layers are recursively generated imposing a fixed correlation with the previous ones. In details, the timestamp $n+1$ is obtained  by randomly swapping a fixed fraction $f$ of couples of edges present in the network at timestamp $n$. In this way the temporal network that we obtain is characterized by a realistic temporal correlation between timestamps and each static network has the same degree distribution. We chose $f=0.3$ and we used as initial static networks six different  graphs: two \ER~\cite{erdos1960evolution} networks (with $p=0.01$ and $p=0.001$), two scale-free networks~\cite{barabasi1999emergence} (with the same parameters $\alpha=0.41$, $\beta=0.54$, $\gamma=0.05$, $\delta_{in}=0.2$, $\delta_{out}=0$ of the algorithm described in \cite{bollobas2003directed} but with two different random seeds), and two  small-world networks~\cite{watts1998collective} (with $p=2$ and $p=8$, but $k=3$ for both).
}
\review{
 Table~\ref{tab:synthetic_graphs} shows the parameters of the generated graphs.
\begin{table}[h!]
    \begin{center}
    \begin{tabular}{ l| c| c| c}
     Name & \# nodes & \# edges & Time stamps \\     
     \hline\hline
     Erdos Renyi (p=0.01) & 64 & 1610 & 301 \\     
     Erdos Renyi (p=0.001) & 13 & 78 & 301 \\     
     Scale Free (G1) & 100 & 2748 & 301 \\     
     Scale Free (G2) & 100 & 3524 & 301 \\     
     Small World (p=2) & 100 & 4581 & 301 \\     
     Small World (p=8) & 100 & 4340 & 301
    \end{tabular}
    \caption{Synthetic temporal graphs parameters}
    \label{tab:synthetic_graphs}
    \end{center}
\end{table}
}

\subsection{Non-egocentric miners}
\label{sec:tmm}
As previously stated, no alternative approaches exist that focus on
mining egocentric temporal motifs. However, to provide some
comparative evaluation for the results of our mining algorithm, we
also ran the state of the art non-egocentric temporal motif mining
algorithm by Paranjape \emph{et al.} \cite{paranjape2017motifs}. 
Note however that the motifs found by this method are prototypical of what any non-egocentric mining approach can produce. As mentioned in the related work, the method can be  described as follows: $(i)$ obtain the aggregate graph of the input temporal graph (see Definition~\ref{def:tg}); $(ii)$ extract (static) $n$-node $l$-edges motifs, where $n$ is the number of nodes in the motif and $l$ is the number of edges (parameters of the algorithm), using standard approaches for determining motifs (where the null models have the same aggregate degree distribution of the input graph); and $(iii)$ for each static motif count its isomorphic sub-graphs on the temporal network, i.e. with edges possibly appearing at different times. If the maximum distance in time among the different edges is less than a given time $\delta$, the sub-graph is denoted as a temporal motif. In the following, we refer to the Paranjape \emph{et al.} \cite{paranjape2017motifs} method as \tmm.

\subsection{Non-egocentric graph distances}
\label{sec:nonegodist}
In this subsection, we present four distances based on micro-scale, meso-scale and global features of the temporal graph.\\
\textbf{NetSimile:} Berlingerio \emph{et al.}~\cite{berlingerio2012netsimile} developed NetSimile, a tool for network distance. This method relies on a set of seven features of the network's nodes. Such features are: degree of the nodes, clustering coefficient, average number or nodes in two-hop neighborhood, average clustering coefficients of the neighbors of a node, number of edges in the node egonet (induced sub-graph of node and neighbors), number of outgoing edges and number of neighbors of the ego. First, the median, mean, standard deviation, skewness, and kurtosis are computed for each feature, producing a graph embedding of $7 \times 5 = 35$ elements. Then, the distance among graphs is computed as the Canberra distance between their respective embeddings. \\
To apply such method to the aforementioned datasets, we compute the aggregated network, that is, the network obtained by removing the temporal dimension in the input data and the duplicated edges.\\
\textbf{Modified NetSimile:} NetSimile is not originally conceived for temporal graphs. We thus considered a variant of the method that includes the number of temporal interactions of a node as an additional feature over which to compute the statistics, thus producing an embedding of dimension $40$. \\
\textbf{Weighted Laplacian:} While previous distances rely on local features of the input graph, the Weighted Laplacian leverages global features. First of all, a weighted aggregated static graph is created, in which the weights on an edge represent the number of interactions (over time) that the edge has had. Then the Laplacian matrix is defined as $L = D - W$, where $D$ is the degree matrix and $W$ is the matrix of edge weights.\\
To compute the distance among two temporal graphs $\GT_1$ and $\GT_2$, we calculate the Laplacian matrices $L_1$ and $L_2$, then we set $k$ equal to the minimum number of nodes between $\GT_1$ and $\GT_2$, and finally we compute the Euclidean distance between the first $k$ eigenvalues of $L_1$ and $L_2$.\\
\textbf{Temporal motifs:} 
To compute the distance between networks using meso-scale features, we considered a distance induced by (non-egocentric) temporal motifs. This is achieved by applying a variant of Definition~\ref{def:etm_dist} that uses temporal motifs as discussed in Section~\ref{sec:tmm} in place of ETM.

\section{Results}\label{results}

We start by showing qualitative results in which we compare egocentric and non-egocentric motifs, and then report a  quantitative analysis of the effectiveness of our ETM-based graph distance as compared to alternative non-egocentric graph distance measures.

\subsection{Egocentric vs non-egocentric temporal motifs}

We compare motifs found by our \etmm with those generated by \tmm. We set the number of temporal steps $k=2$ for \etmm, while for \tmm we consider $3$-nodes and $3$-edges structures. These values allow to generate non-trivial motifs and to find a significant amount of them in each dataset.
As will be clear in the next, the difference between the methods is evident and does not depend on the specific choice for these parameters.  Note that \etmm does not require to set the number of nodes and edges and it can in principle extract motifs with an arbitrary number of neighbors. Following Milo \emph{et al.}~\cite{milo2002network} we set the number of null models $n=100$, with parameters $\alpha=0.01$, $\beta=0.1$ and $\gamma=5$.   
\begin{figure}[h!]
    \centering
    \includegraphics[width=\linewidth]{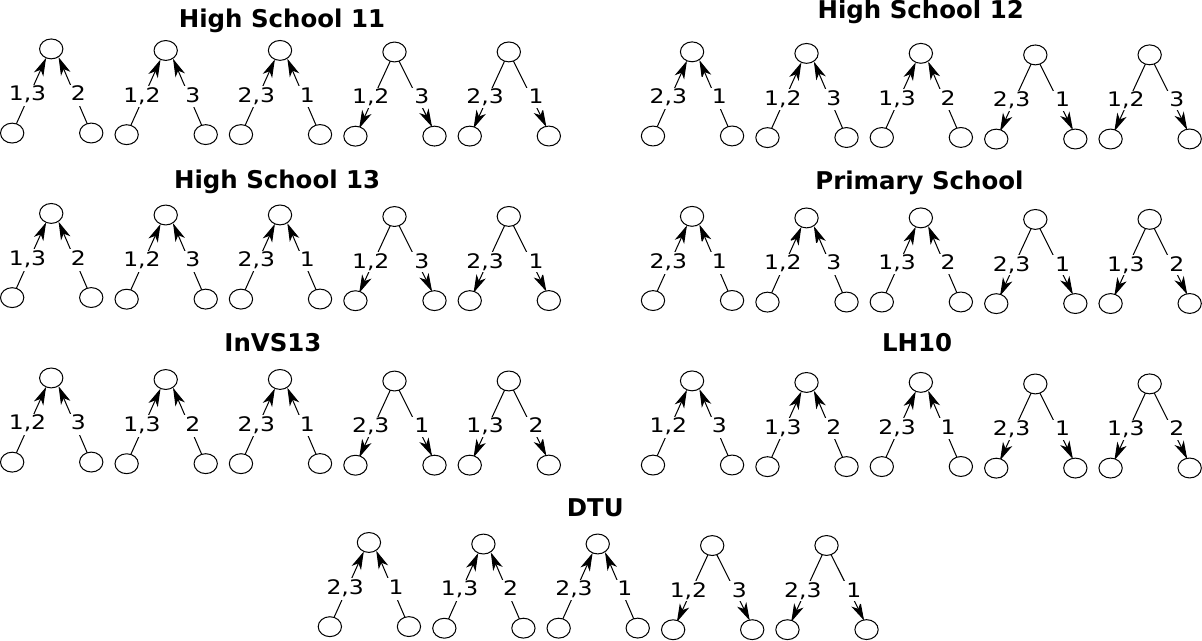}
    \caption{Most frequent temporal motifs discovered by \tmm on the seven networks for $\dt=300$.}
    \label{fig:tmm_300}
\end{figure}

To give an insight of the main differences between egocentric and non-egocentric motifs and highlight the usefulness of the former in discovering patterns of social interaction, we report the five most frequent motifs found by the different methods. We focus on a temporal gap $\dt = 300$ seconds, but results are quite similar for different temporal gaps.

Figure~\ref{fig:tmm_300} shows the first five motifs found by \tmm on the different datasets, ordered by frequency. These motifs show some dynamics in the interaction, but it is difficult to
interpret them in terms of social interaction patterns or to identify
some clear features that distinguish the various datasets. Moreover, Figure~\ref{fig:tmm_300} shows that the five most frequent motifs are the same for all the datasets, with the only exception of the fifth motifs of \textit{InVS13} and \textit{LH10}.

\begin{figure}
    \centering
    \includegraphics[width=0.9\linewidth]{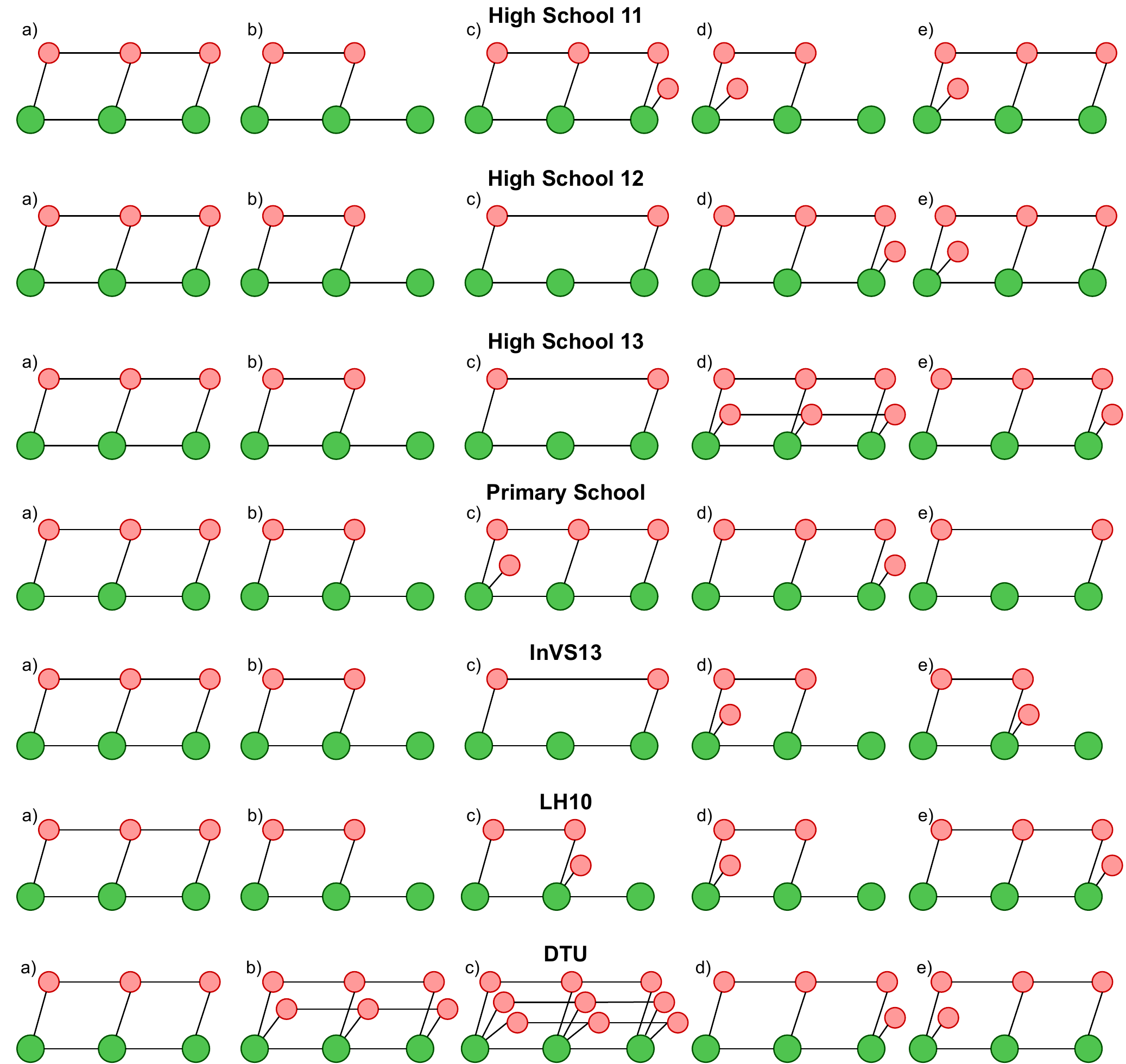}
    \caption{Most frequent egocentric temporal motifs discovered by \etmm  on the three datasets for $\dt=300$.}
    \label{fig:etmm_300}
\end{figure}

The five most frequent motifs discovered by our method are reported in Figure~\ref{fig:etmm_300}. Note that the egocentric focus allows to generate motifs which are quite interpretable in terms of social interactions of the person under investigation (the ego). For instance, for \textit{HighSchool11} (first line) we identify a continuous interaction with another person (a) (c) (e), possibly combined with a third person joining at the beginning (e) or at the end (c) of the interaction.

Concerning the other datasets, even if the first two ETMs are the same (except for the \textit{DTU} dataset), our approach does identify some differences that can be related to the different type of networks under investigation. For example, our method is able to identify motifs characterized by rich and dynamic interactions among students in high school
and university, and by sparse and short interactions among employees in the research institute.
The last line of Figure~\ref{fig:etmm_300} shows the motifs found by \etmm on the \textit{DTU} dataset, and it is easy to see that the structures of the discovered motifs are quite different and more complex with respect to the structures of the motifs found in the other datasets. This may also depend on the fact that the \textit{DTU} dataset, collected using Bluetooth technology, captures co-location and not face-to-face interactions. 

\begin{figure}
    \centering
    \includegraphics[width=0.6\linewidth]{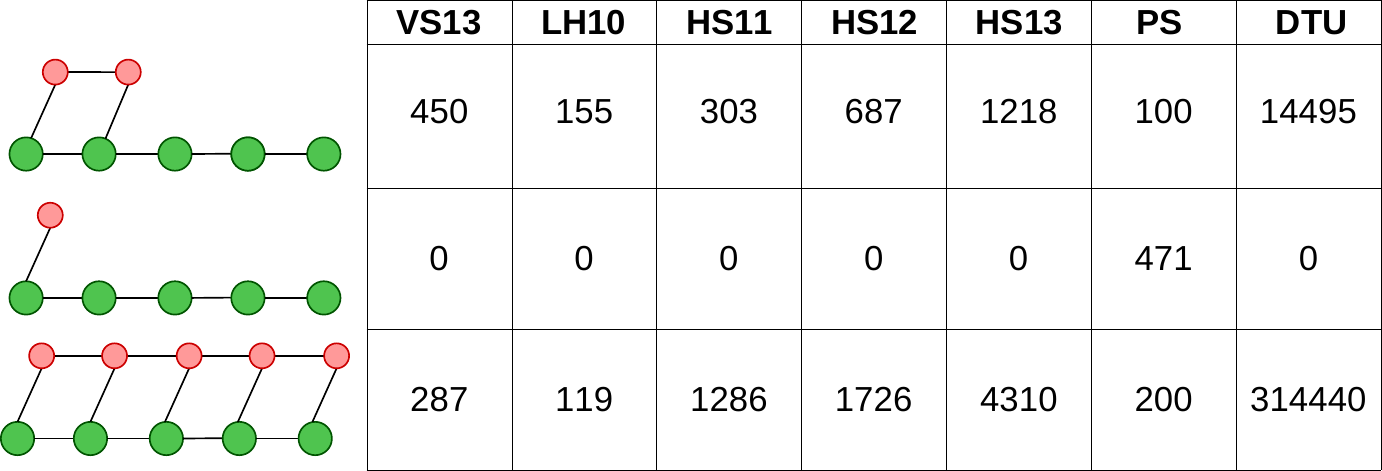}
    \caption{The figure shows the frequencies of three ETM for each dataset. The ETMs are those with the maximum variance among datasets.}
    \label{fig:solo3_etn_count}
\end{figure}

\review{
To provide further insights on the relationship between motifs and types of networks, we looked for the set of temporal motifs that most discriminates among different datasets. We selected the three  egocentric temporal motifs with maximum variance of occurrence among the datasets, and report their frequencies in Figure \ref{fig:solo3_etn_count}. The difference between the primary school and the other datasets is striking. The former contains a motif that is totally missing in the other networks, namely the case where an individual briefly interacts with another one (for less than 5 minutes) and has no more interactions in the following 20 minutes. This small set of motifs may seem a poor description of the analyzed social settings. However, it is surprisingly accurate in catching differences and similarities among datasets, as we will see in next section.
}

\subsection{Egocentric vs non-egocentric graph distances}
To give some quantitative estimate of the descriptiveness of the motifs found by our method, we study their effectiveness in measuring the distance among the networks described in Section \ref{Experimental_evaluation}. In particular, we show the importance of the egocentric perspective in identifying similar social contexts by means of network distances (Definition~\ref{def:etm_dist}).\\

\begin{figure}[h!]
    \centering
    \includegraphics[width=\linewidth]{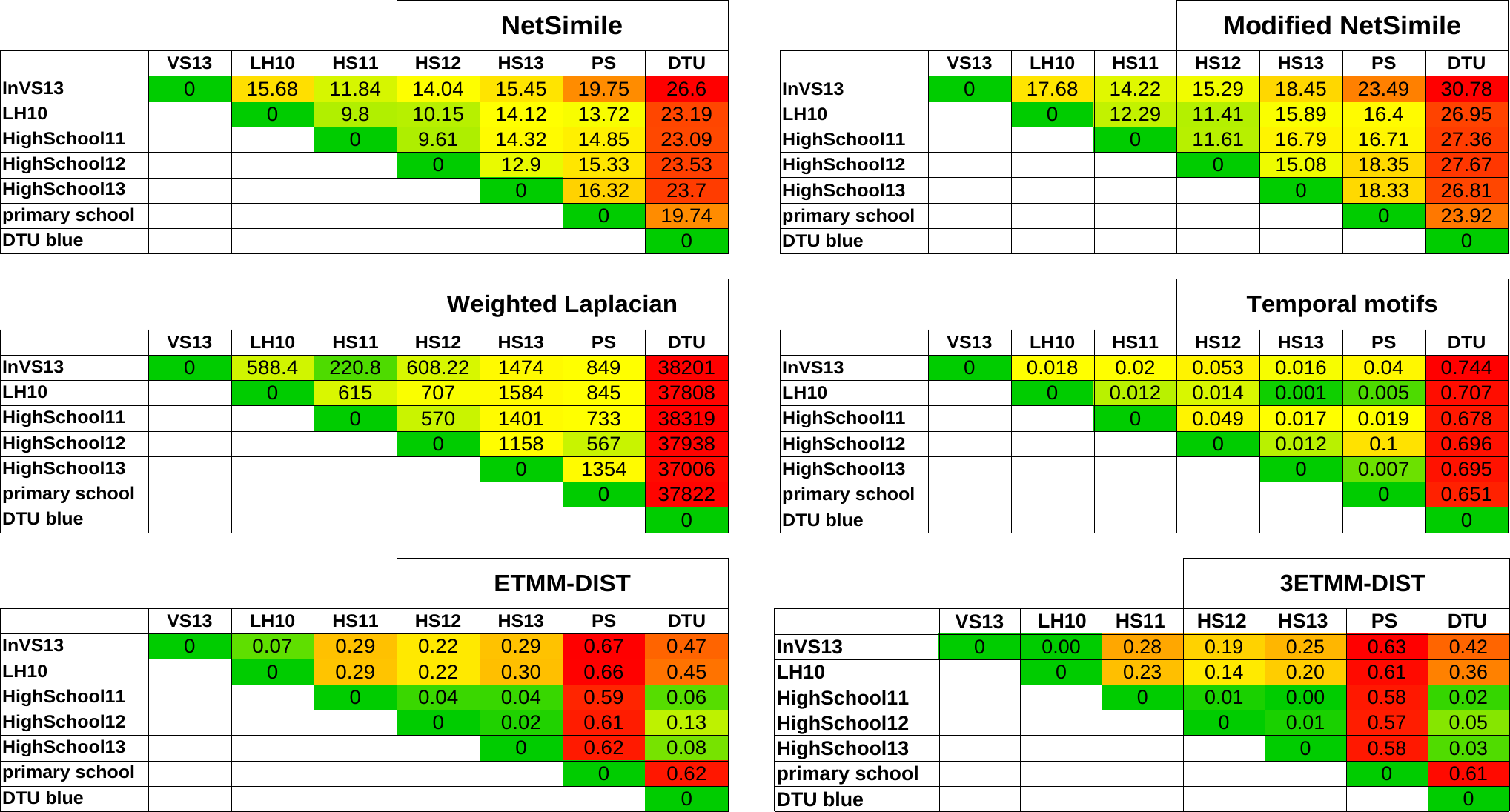}
    \caption{Distances among networks using five different methods, namely: NetSimile, Modified NetSimile, Weighted Laplacian, Temporal motifs and ETM-based distance. Each element of the table is colored in a color scale from green (minimum distance) to red (maximum distance). \review{The last tables show the distances obtained using ETM-based distance using $\dt = 300$ and $k=4$, with all motifs (left) and the three most discriminative ones (right)}.}
    \label{fig:distances_competitors}
\end{figure}

Figure \ref{fig:distances_competitors} shows the distances computed with the four non-egocentric methods reported in Section~\ref{sec:nonegodist} (first two rows) and with our \review{ETM-based distance (last row)}. Each table reports the pairwise distances between networks and each element is colored with a color scale starting from green (minimum distance) to red (maximum distance).
The figure clearly shows that all non-egocentric methods have serious problems in producing meaningful distances between interaction networks. First, all of them consider \textit{DTU} to be the farthest away from all other networks. However, we expect that \textit{DTU} network, which collects the co-location behaviors of university students, should show some similarities with the ones capturing the face-to-face interactions of high school students, namely \textit{HighSchool11}, \textit{HighSchool12} and \textit{HighSchool13}. These similarities seem not adequately detected by these methods. As previously anticipated, the fact that \textit{DTU} results appear so different from those obtained with the other datasets may depend on the fact that different technologies (RFID vs Bluetooth) have been used to collect the data, suggesting that non-egocentric approaches fail in revealing social patterns when different technologies are at place.
We also notice that both NetSimile and Modified NetSimile (first row) detect \textit{hospital} (\textit{LH10}) as the closest network to \textit{primary school}; this appears as an unexpected result, considering the differences between these two social contexts. 
Moreover, the Weighted Laplacian method (first table second row) fails in identifying similar environments, since we observe that \textit{InVS13} is very close to \textit{HighSchool11} but quite distant from \textit{HighSchoo13}.
Finally, according to (non-egocentric) temporal motifs (second table second row)
the network \textit{LH10} is very close to almost all datasets, being almost identical to \textit{HighSchool13}. 

\review{The last row of Figure \ref{fig:distances_competitors} shows the results of our ETM-based distance (for $\dt=300$ and $k=4$), using all ETMs (left) and only the three most discriminative ones (right), i.e., those maximizing the variance of ETM frequencies among datasets (shown in Figure~\ref{fig:solo3_etn_count}).}
The reported network distances provide a more satisfactory description of the similarity between the underlying datasets. First of all, the three high school networks are very close to each other, with distances around 0, while presenting larger distances with all other networks. Moreover, among the other networks, the closest one is represented by the one capturing the co-location behavior of university students (\textit{DTU}), which are expected to share some behavioral routines with high school students (e.g., class attendance). This shows that ETM is capable of finding similar social interaction patterns despite the use of different data collection technologies, something alternative non-egocentric measures completely fail to achieve. The network which is farthest away from all the others is the \textit{primary school} network: this may be explained by the fact that primary school children seem to experience interaction dynamics which are significantly different from the ones characterizing the social settings of young adults and adults.
Finally, we observe that another sensible niche is represented by the two working places, namely the hospital and the research institute, quite similar between each other and quite distinct from all other settings.
\review{Interestingly, limiting the set of ETMs to the three most discriminative ones produces results which are very similar to those obtained with the full set of motifs (around six thousands). This is a surprising result and a confirmation of the effectiveness of the egocentric perspective in characterizing different types of social interaction settings.} 

\subsection{Sensitivity analysis}

\review{
In the following we provide a sensitivity analysis showing how the choice of the parameters, namely the temporal gap $\dt$ and the temporal neighborhood order $k$, affect the ETM-based distance. In Figure~\ref{fig:sensibility} we report the ETM-based distance among datasets using $\dt$ equal to $300$
and $900$ seconds\footnote{A value of $\dt < 300$ generates a too sparse network for the DTU dataset that relies on Bluetooth to detect interactions, preventing the discovery of non-trivial motifs by any method. Results for the other datasets are similar for values of $\dt$ as small as 60.}, and $k$ ranging from $3$ to $5$. We observe that results are quite stable. For intermediate values of the parameters, results are very similar to those presented for $\dt=300$ and $k=4$, with distances that tend to increase for increasing values of $\dt$ and $k$. Intuitively, small values of {\em both} $\dt$ and $k$ (i.e., $\dt=300$ and $k=3$, top left matrix) produce small motifs, leading to a partial reduction in discriminative capacity, with the primary school becoming (more) similar to workplaces and high schools. On the other side, large values of {\em both} $\dt$ and $k$ (i.e., $\dt=900$ and $k=5$, bottom right matrix) determine a slight decrease in the capacity of detecting similarities among related datasets (namely between different high schools). This is again not surprising, as jointly increasing $\dt$ and $k$ substantially increases the required length for a temporal fragment to match a motif, making it more complex for the method to mine relevant motifs.}

\begin{figure}
    \centering
    \includegraphics[width=\linewidth]{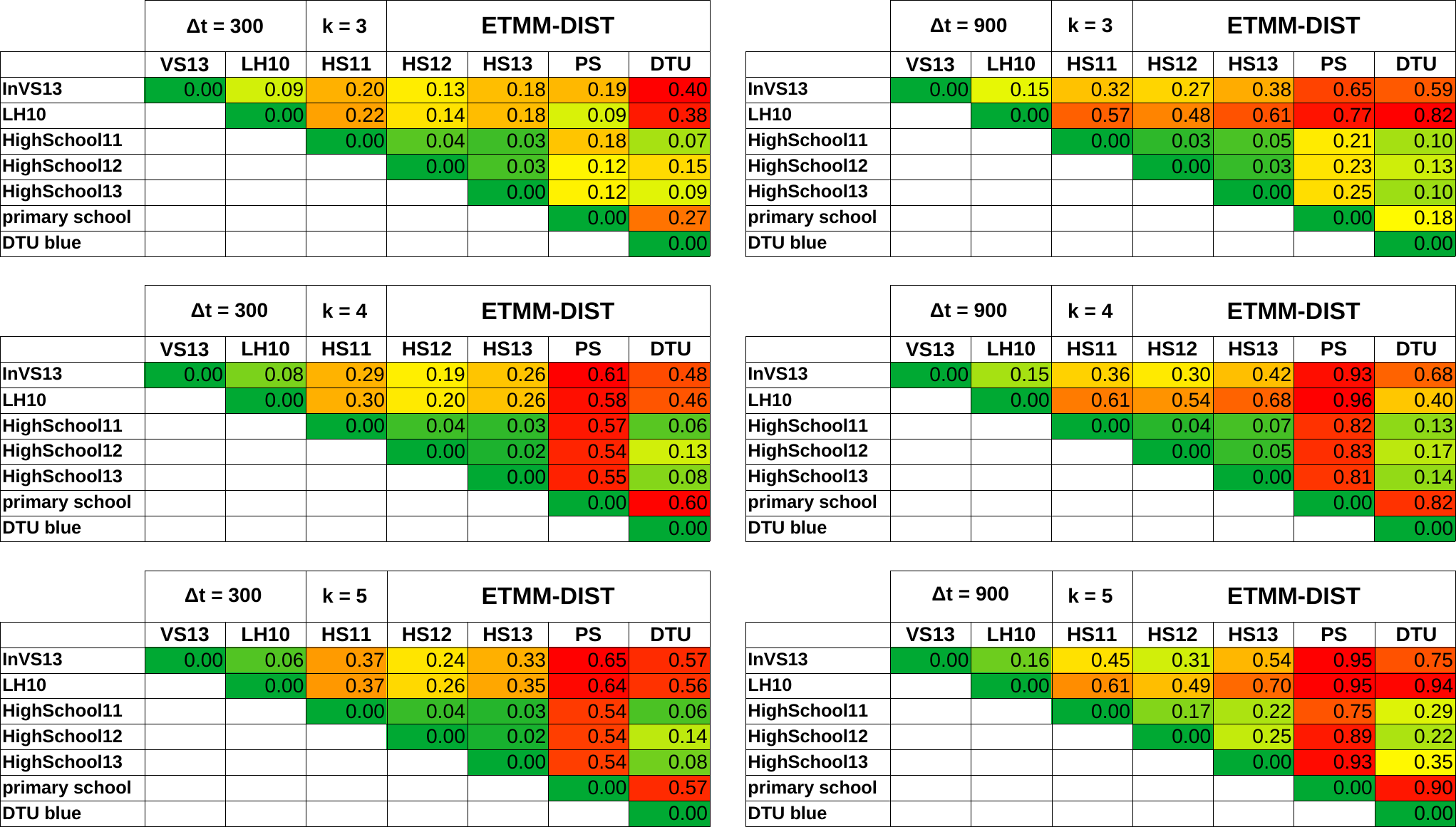}
    \caption{ETM-based distances obtained using $\dt= 300,900$ and $k=3,4$ and $5$.}
    \label{fig:sensibility}
\end{figure}

\review{
\subsection{Results on distance communication and synthetic datasets}
}

%
\review{%
In this section we evaluate the ability of the ETM-based distance to characterize networks beyond close proximity interaction data. First of all we consider other typologies of social data representing distance communication interactions (Fig.~\ref{fig:altri_datasets}), then we explore the  algorithm performance on synthetic temporal graphs (Fig.~\ref{fig:altri_datasets_sintetici}).
}

\begin{figure}[h!]
    \centering
    \includegraphics[width=\linewidth]{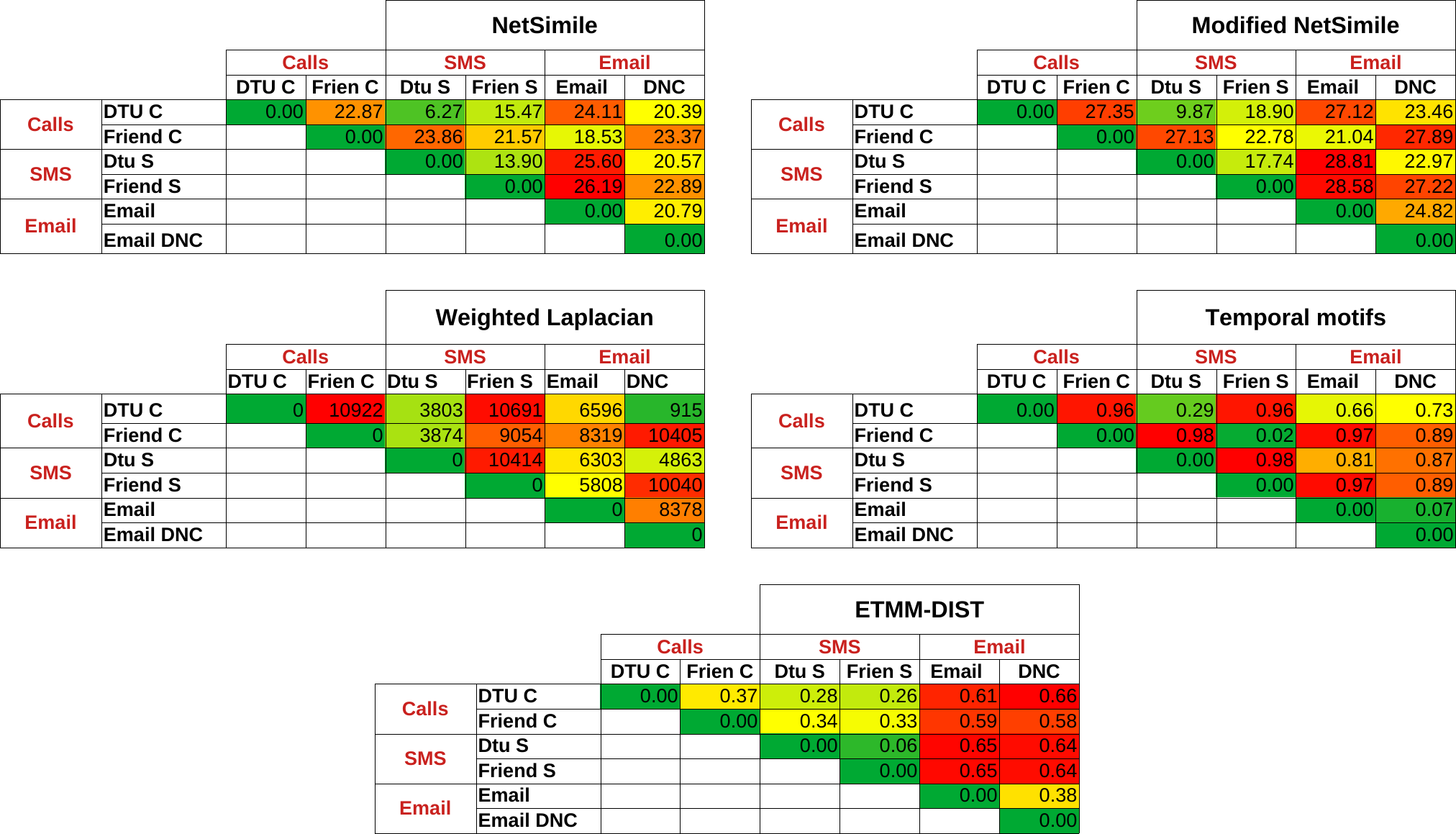}
    \caption{\review{Distances among different communication networks using five different methods, namely: NetSimile, Modified NetSimile, Weighted Laplacian, Temporal motifs and ETM-based distance. Each element of the table is colored in a color scale from green (minimum distance) to red (maximum distance). The last table shows the distances obtained using ETM-based distance with $\dt = 3600$ and $k=4$.} 
    }
    \label{fig:altri_datasets}
\end{figure}

\review{
The non-physical interaction datasets that we consider employ different communication technologies (phone calls, SMSs and emails, see Section~\ref{sec:nonphysical_data}). For this experiment, we choose $k=4$ and $\dt=3600$, a temporal gap for which the six temporal networks are characterized by a similar average degree (equal to 0.052 for phone calls, 0.049 for SMSs and 0.051 for emails). Results are shown in Figure~\ref{fig:altri_datasets}. Non-egocentric methods manage to capture the similarity among some of the networks using the same technology (e.g., SMSs for NetSimile and Modified NetSimile, emails for temporal motifs), but they badly fail in most cases. On the other hand, the ETM-based distance is quite consistent in capturing the similarity between networks employing the same communication technology. Moreover, networks based on SMSs and phone calls are more similar to each other than networks based on emails, as expected. This is a further proof of the versatility of ETM patterns to characterize temporal behaviors.}

\begin{figure}[h!]
    \centering
    \includegraphics[width=\linewidth]{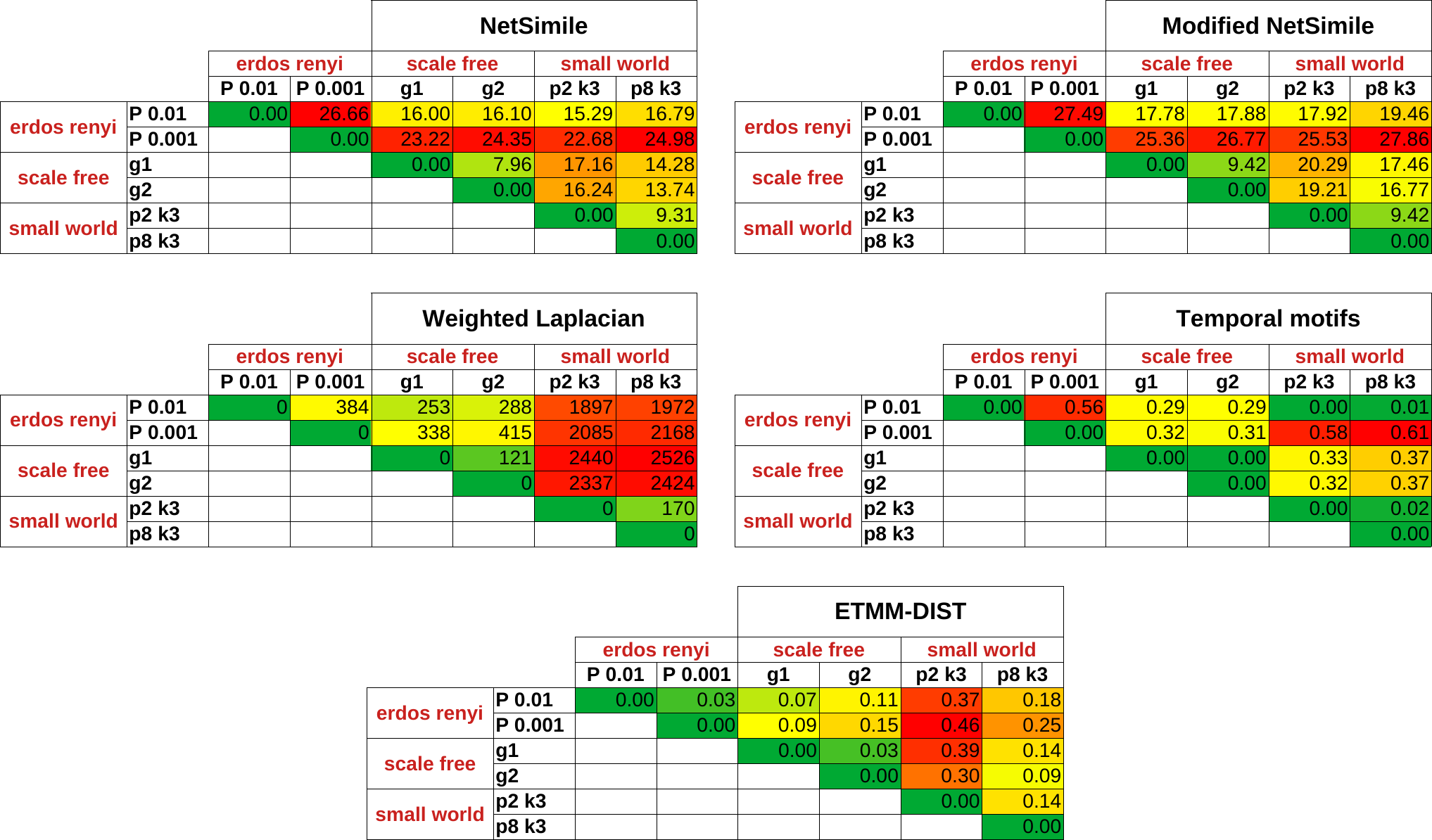}
    \caption{\review{Distances among different synthetic networks using five different methods, namely: NetSimile, Modified NetSimile, Weighted Laplacian, Temporal motifs and ETM-based distance. Each element of the table is colored in a color scale from green (minimum distance) to red (maximum distance). The last table shows the distances obtained using ETM-based distance with $k=4$. No temporal gap is needed in this case, as the networks are not extracted from time series but synthetically generated according to the procedure described in Section~\ref{sec:synth_data}.}
    }
    \label{fig:altri_datasets_sintetici}
\end{figure}

\review{
Results for the synthetic datasets are shown in Figure \ref{fig:altri_datasets_sintetici}. Our ETM-based distance is clearly capable of detecting similarities among \ER graphs, outperforming all competitors, and among scale-free ones, which are however modelled reasonably well by all methods. On the other hand, the ETM-based distance lags behind all competitors in detecting similarities between small-world networks. This result sheds some light on the limitations of the egocentric perspective of our method. Indeed, ETMs deliberately discard the information on connections among neighbors of the ego node (we only consider the existence of neighbors and not their mutual behavior), thus neglecting the clustering structure of the network. This explains why the synthetic small-world networks, characterized by high values of clustering coefficient, are less effectively described by our method.}

\section{Conclusion}\label{conclusion}

In this manuscript we proposed a novel approach for mining temporal
motifs based on an ego perspective. Each motif represents the
evolution, during few time steps, of the set of neighbors of a
specific network node. Egocentric temporal motifs present some
essential characteristics that distinguish them from standard temporal motifs.

First of all, egocentric temporal motifs are simpler, at a topological level, with respect to standard temporal motifs, since they only take into account the neighboring nodes of the ego, ignoring the connections among them. This allows both to account for larger neighborhoods and to explore more in detail the temporal aspect, including duration of contacts and contemporary interactions, usually neglected in standard procedures for temporal motif mining. This is a necessary requirement when analyzing social domains like physical human interactions, where each individual can interact with multiple people at a time, with various durations.

Second, the egocentric view has substantial advantages from a computational perspective. Traditional techniques for motif mining rely on an isomorphism test for assessing if two sub-networks are equivalent or not, and this limits their applicability to mine motifs containing a handful of nodes. The focus on an ego node allows us to sidestep this problem. We show how an egocentric temporal neighborhood, which is the sub-structure representing a candidate motif, can be encoded into a bit vector in a way such that two neighborhoods have the same encoding if and only if they are isomorphic.

\review{We made use of seven different datasets representing social interactions and applied our egocentric temporal motif miner, comparing the results with a state of the art non-egocentric temporal motif miner. Our method is shown to be more effective in terms of selectivity and quality of the extracted motifs. By visually inspecting the most frequent motifs found in each dataset, it is apparent that our method succeeds in grasping some of the peculiarities of each dataset: more rich and dynamical interactions among students in high school and university, sparser and shorter interactions for the research institute, a combination of the two in the hospital, and a different behavior at the primary school. 
Importantly, differences and similarities between datasets are quantified by defining a correlation measure between egocentric signatures. The results that we obtain fully reflect the social context represented by the network, especially if compared with standard non-egocentric approaches to measure temporal networks' distance. Later, we show how the egocentric perspective is crucial for the discrimination among different communication technologies, like phone calls, SMSs and emails, and how it also allows to characterize temporal variants of popular network topologies, like \ER and scale-free.
}

\review{The egocentric perspective surely represents an important limitation too, since we are neglecting all the second order interactions, i.e., the interactions between neighbors of an ego node. This is especially limiting in networks which are characterized by a high clustering coefficient, as shown by the suboptimal results that we achieve on small-world networks. On the other hand, this is a necessary requirement for the bit vector encoding and hence for the extreme velocity of our method (which scales linearly with the number of nodes and the timesteps of the temporal network). This allows to mine motifs covering larger structures and longer time sequences with respect to alternative solutions. Our extensive experimental results show that, even renouncing to represent second-order interactions, the proposed method is able to recognize different social settings, substantially outperforming existing alternatives.
}
\\
In conclusion, we are proposing a novel efficient method to obtain temporal motifs from the node point of view. This method is not conceived to completely replace existing temporal motif mining methods, but rather to complement them in revealing a different kind of motifs. As shown in our experimental evaluation, this can be particularly useful to study social interaction networks, which could not be properly analyzed with existing approaches.
\\ \\
\textbf{Data availability:} The source code is freely available at:\\ https://github.com/AntonioLonga/Egocentric-Temporal-Motifs-Miner-ETMM. The datasets employed in our experiments can be downloaded at:
\begin{itemize}
    \item http://www.sociopatterns.org (sociopattern datasets)
    \item https://doi.org/10.6084/m9.figshare.7267433 (DTU)
    \item http://realitycommons.media.mit.edu/friendsdataset.html (Friends\&Family)
    \item http://snap.stanford.edu/data/email-Eu-core-temporal.html (Emails)
    \item http://networkrepository.com/email-dnc.php (Emails DNC)
\end{itemize}

\bibliographystyle{acm}
\bibliography{bib_sociopattern}

\end{document}